\documentclass[reqno]{amsart}
\usepackage{amsmath}
\usepackage{amsfonts}
\usepackage{amssymb}
\usepackage{graphicx}
\usepackage{amsthm}
\usepackage{color}

\numberwithin{equation}{section}

\theoremstyle{plain}
\newtheorem{theorem}{Theorem}[section]
\newtheorem{proposition}[theorem]{Proposition}
\newtheorem{lemma}[theorem]{Lemma}
\newtheorem{corollary}[theorem]{Corollary}
\newtheorem{remark}[theorem]{Remark} 
\newtheorem{definition}[theorem]{Definition}

\newtheorem{hypothesis}[theorem]{Assumption}

\renewcommand{\P}{\mathbb{P}}
\newcommand{\E}{\mathbb{E}}
\newcommand{\R}{\mathbb{R}}
\newcommand{\F}{\mathcal{F}}
\newcommand{\FF}{\mathbb{F}}

\title{Applications of a New Self-Financing Equation}
\date{\today}
\author{Ren\'e Carmona and Kevin Webster}
\address{ORFE, Bendheim Center for Finance\\
Princeton University}

\begin{document}

\begin{abstract}
The goal of this note is to illustrate the impact of a self-financing condition recently introduced by the authors. We present the analyses of two specific applications usually considered in more traditional models in financial mathematics. They include hedging European options with limit orders and the optimal behavior of market makers.
\end{abstract}

\maketitle

\section{Introduction: the self-financing equation}

In quantitative finance, the standard self financing equation is a cornerstone of the theory of frictionless markets. It plays a crucial role in many fundamental results. Mathematically, speaking it is a simple equation which \emph{constrains} the wealth process of an investor to live in a certain sub-space. This sub-space is therefore viewed as a space of \emph{admissible} portfolios. 
New-comers to the mathematical theories of financial markets often gripe with the self-financing condition and how it relates to the real world. While it can be postulated as a mathematical definition, it can also be \emph{derived} from a limiting procedure starting from an accurate description of the microstructure of trades in the trade clock. This approach is at the core of our strategy, and to implement it, we shall have to overcome the idiosyncrasies of the passage from discrete to continuous time.

\begin{quote}
``The sad fact is that the self-financing condition is considerably more subtle in continuous time than it is in discrete time.''\footnote{J. Michael Steele, \emph{Stochastic Calculus and Financial Applications}, section 14.5 'Self-financing and self-doubt'.}
\end{quote} 

When discussing market models at the macroscopic level, we assume that the mid-price $p$ and the inventory $L$ are given by It\^o processes:
\begin{equation}
	\begin{cases}
		dp_t &= \mu_t dt + \sigma_t dW_t \\
		dL_t &= b_t dt + l_t dW'_t
	\end{cases}
\end{equation}
for two Wiener processes  $W$ and $W'$ with unspecified dependence structure. In the simplest case, we also consider an adapted process $s_t$ acting as a proxy for what remains in the continuous time limit of the bid-ask spread measured \emph{in tick size}.
The standard self-financing condition of continuous time finance can be stated as a constraint:
\begin{equation}
\label{fo:csf}
dX_t = L_t dp_t
\end{equation}
between the price $p$ of the underlying interest, the inventory $L$, and the wealth $X$ of the agent.
In most classical financial models, Merton's portfolio theory is a case in point, the price $p$ is exogenously given, the inventory $L$ is the agent's input, and her wealth $X$ appears as the output of equation \eqref{fo:csf}. 

The objective of this paper is twofold. We first formalize the order book in a mathematical fashion and derive the associated transaction costs and trade equations. The second goal is to generalize the self-financing portfolio condition \eqref{fo:csf} to incorporate known pecularities of the high frequency markets including transaction costs, price impact and price recovery. Among other things, we want this generalization to be able to quantify the differences between trading via limit orders and market orders.
Finally, we want to warn the reader that the equations proposed in this paper are only \emph{necessary}, and that quantifying limit order fill rates, priorities and price recovery are beyond the  scope of the present paper.

\subsection{The order book}

We introduce the order book first as a pair of positive measures $(b,a)$ on the price grid. Under the assumption that the mid-price is well-defined, an equivalent definition in terms of an order book \emph{shape function} $\gamma$ is introduced. Formally, $\gamma'' = a + b$ where $\gamma''$ is the second derivative of the function $\gamma$ in the sense of Schwartz distributions.

Transaction costs are shown to be given by the Legendre transform $c$ of the order book shape function $\gamma$. This leads to explicit formulas for all mechanical transactions on the order book, such as the instantaneous price impact, the traded volume, etc. In particular, the discrete time equation for the  wealth associated to a self-financing portfolio will be shown to be:
\begin{equation}
\Delta X = L \Delta p \pm c(\mp \Delta L) +\Delta p \Delta L
\end{equation}
where $\pm$ is $+$ when trading with limit orders and $-$ when trading with market orders.

\subsection{Proposed self-financing equation}
In continuous time, the corresponding self-financing equation takes the form:
\begin{equation}
\label{eq:Intro_Formula}
dX_t = L_t dp_t \pm \int_\R c(y) \phi_{\sigma^2_t}(y) dy dt +d[L,p]_t 
\end{equation}
where as before $\pm$ is $+$ when trading with limit orders and $-$ when trading with market orders, and $\phi_{\sigma^2}$ is the density function of the 
Gaussian distribution with mean $0$ and variance $\sigma^2$.
We show in Section \ref{sec:order_book} below that, when time is measured in the trade clock, the discrete time analog of formula \eqref{eq:Intro_Formula} can be derived rigorously from a specific limit order book feature. It also matches real wealth data (see \cite{cw_self_financing} and the appendix at the end of the paper for empirical evidences).
We shall also impose the constraint 
\begin{equation}
\label{fo:corr_constraint}
d[L,p] < 0
\end{equation}
whenever trading is done with limit orders. The interpretation for this constraint is price impact. It has also been thoroughly tested on high frequency data in \cite{cw_self_financing}. See also the appendix at the end of the paper for evidence from market data from a different exchange.
\vskip 2pt
We now explain how our condition \eqref{eq:Intro_Formula} and the adverse selection constraint \eqref{fo:corr_constraint} relate to the conditions most often found in the literature. Later on in Section \ref{sec:diffusion_limit}, the latter will be derived as continuous time limits of discrete self-financing equations under different scaling assumptions. 

\subsection{The Almgren-Chriss model}
The seminal work of Almgren and Chriss \cite{Almgren} addresses a  closely related question. These authors propose a \emph{macroscopic model} for the price impact and the change of wealth after a liquidity taker's decision. The model leads to a very tractable framework which was, and still is, used in many optimal execution studies (see \cite{Alfonsi, Wang} for example). This framework can be summarized by the system:
\begin{equation}
	\begin{cases}
		dp_t &= f(l_t) dt + \sigma_t dW_t \\
		dL_t &= l_t dt \\
		dX_t &= L_t dp_t -c(l_t)dt
	\end{cases}
\end{equation}
where $f$ and $c$ are positive functions\footnote{$c$ should be understood as a transaction cost function. For this reason, it is often assumed to be convex.}.

The main advantage of this model is that price impact appears in a tractable fashion. Indeed, it comes through the drift $f(l_t)$ of the price process, which creates a \emph{positive correlation} between traded volumes and the price. However, it constrains $L$ to be a \emph{differentiable} function of time, and as a result, the model parameters cannot be calibrated to market data directly, making the model difficult to test empirically. 
As per the empirical analysis of NASDAQ data reported in \cite{cw_self_financing}, there is ample evidence supporting nondifferentiable inventories (see also the appendix at the end of the paper). Moreover, certain trading strategies, such as delta-hedging, latency arbitrage, and statistical arbitrage, naturally lead to inventory models with infinite variation. Finally, the use of limit orders is not covered by the Almgren-Chriss approach.

\subsection{Transaction cost literature}
The branch of classical mathematical finance most related to our paper is portfolio selection under transaction costs (\cite{ Chellathurai, Magill, Shreve} or the recent review \cite{MuhleKarbe}). Most of these works start from  an equation for the wealth of a liquidity taker which generalizes the classical self-financing equation to a setting with transaction costs. In general however, these papers do not underscore the derivation of the model, but instead, they emphasize the study of its consequences. We hope to appeal to this part of the community by providing more accurate equations for self-financing portfolios while keeping reasonable tractability, leading the way to problems related to \emph{liquidity provision}, such as market making. An interesting feature of such applications is that the agent does \emph{not} directly control her portfolio, adding an additional modeling challenge. For the record we note that the standard equation used in this branch of the literature is:
\begin{equation}
\label{eq:Intro_standard}
dX_t = L_t dp_t - \frac{s_t}{2} |dL|_t 
\end{equation}
where $s_t$ is the bid-ask spread, and again, the inventory process $L$ is assumed to have finite variation, i.e. $\int_0^t |dL|_s <\infty$ for all finite $t$.

Strengths of this model are its simplicity, relative tractability, and straightforward calibration to the market.
However, we see a shortcoming in the fact that the process $L$ can only have finite variation\footnote{See \cite{cw_self_financing} and the appendix for why this is problematic.}. Moreover, price impact, limit orders and other microstructure considerations are absent from the model.

\subsection{Methodology}
Rather than postulating the definition of a self-financing portfolio directly in the continuous limit, our objective is to derive the self-financing equation from a more fundamental perspective. To obtain our results, we therefore propose the following novel strategy:
\begin{enumerate}
\item Define a discrete time representation of the market where all the relevant primary quantities (e.g. price) are rigorously defined on a trade-by-trade basis. This is what we call the 'microscopic' scale.
\item Deduce or define the relevant equations for the derived quantities (e.g. wealth) and trading constraints (e.g. trading via limit orders) on the same scale.
\item Assume that the primary quantities of interest are samples from continuous time diffusion processes, with a sampling frequency that goes to infinity. As a result, each microscopic model is embedded in a sequence of microscopic models that approximate a continuous time model. We call the continuous time limit the 'macroscopic'  limit of the model. 
\item Using appropriate limit theorems, derive the continuous time analogs of the derived quantities and trading constraints.
\end{enumerate}
The main quantity of interest to us, the wealth of a self-financing portfolio, does not require a stochastic model but rather a precise description of the rules underpinning trades on a limit order book. These rules create structural relationships that we exploit to express, at the microscopic scale, wealth as a non-linear functional on the path of our primary quantities. A functional law of large numbers then controls the limiting argument used to derive our continuous time equivalent of the self-financing portfolio equation.

The four-step approach described above can be used in a variety of financial markets. Indeed, it should be possible to derive a self-financing equation for each form of market microstructure. Controlling a continuous time limit for the input variables could then lead to a tractable summarizing equation. While it can hold in other markets, the particularly simple form it takes, and the ease with which it can be tested on empirical data, make high frequency markets an ideal test bed for our theory. Our presentation was influenced by the high frequency markets for three reasons: 1) they are the most relevant among the electronic markets; 2) we can illustrate and test many specific features on empirical data; 3) we believe that our price impact constraint is a particularly strong feature of these markets. 

Our research also casts new light on the difference between trading via limit orders and market orders by modeling price impact. This can be done directly in the continuous time limit, or via the same strategy as for the self-financing equation. In that case, a model for the price move after each trade is proposed on the microscopic scale, and a functional central limit theorem is used to derive the corresponding continuous time equations.

We believe in the potential of our modeling approach, and we hope that other idiosyncrasies of the financial markets will be incorporated into continuous time models through this micro-to-macro approach. Notable papers with a similar microscopic approach are \cite{Stoikov, Cartea, Stoikov2} and \cite{Cont2, Cont}. The former three  papers study microstructure to derive optimal bid-ask spread policies. The latter  two  propose microscopic models of the limit order book. Finally, \cite{Cont3} derives the diffusion limit of one of those microscopic order book models.

\subsection{Basic Results of \cite{CarmonaWebster_FS}}
The first part of the paper formalizes trading on a limit order book, and provides a microscopic description of trades. A duality relationship with transaction costs is introduced and the self-financing condition is derived as a plain accounting relationship for individual trades.

\begin{itemize}
\item We start with a model of the limit order book given by two positive (finite) measures with non-overlapping supports. They represent the distributions of the limit buy and sell orders.
\item From there, we derive the optimal behavior of of a risk neutral liquidity taker: typically he should place an order at the expected value of the future values of the price.
\item We then argue that in discrete time given by the trade clock, the changes in the three fundamental quantities should be linked by the fundamental accounting relationship 
\begin{equation}
\label{fo:discrete_sf}
\Delta X = L \Delta p  + c(-\Delta L) + \Delta p \Delta L 
\end{equation}
which we understand as a single trade \emph{self financing equation}. Here $p$ represents a quoted price, typically the mid-price averaging the best bid and the best ask, $L$ is the inventory (number of shares) and $K$ the amount of cash held by the trader,  and$X$ is the wealth (as marked to the quoted price $P$) defined by $X=pL+K$. As usual we use the notation $\Delta$ for the change of a quantity after a trade takes place. 
\end{itemize}

Then, continuous time versions of the self-financing equation are derived as limits of the discrete self-financing equation under different scaling regimes. We review briefly the various steps of these limiting procedures and we identify one of these limiting models as being more relevant than those currently used in the literature. The back-and-forth approach starts from a continuous time model in which both the quoted price $p_t$ and the inventory $L_t$ are It\^o processes. See \eqref{fo:price_inventory} below. Since the assumption that the inventory process has a non-trivial quadratic variation component is not widely used in the existing financial mathematics literature, \cite{CarmonaWebster_FS} provides empirical evidence in two appendices based on the analyses of the NASDAQ and the TSX high frequency trading data.

Continuous time versions of the single trade self-financing conditions are obtained in the following way. We start from a filtered probability space $\left(\Omega, \F, \FF, \P\right)$ supporting two $\FF$-Wiener processes $W$ and $W'$ with unspecified correlation structure and over a fixed time interval, say $[0,1]$, we give ourselves two $\FF$-adapted processes for the price and inventory of a liquidity provider:
\begin{equation}
\label{fo:price_inventory}
	\begin{cases}
		p_t &= p_0 + \int_0^t \mu_u du + \int_0^t \sigma_u dW_u \\
		L_t &= L_0 + \int_0^t b_u du + \int_0^t l_u dW'_u
	\end{cases}
\end{equation}
where $p_0$ and $L_0$ are square integrable $\F_0$-measurable random variables, and $\mu$, $\sigma$, $b$ and $l$ are $\FF$-adapted and c\`adl\`ag processes. We assume that the order book at time $t$ is given by a convex shape function $\gamma_t$ which is  a continuous $\FF$-adapted process and satisfies $\gamma_t(0)=0$, and we denote by $c_t$ its Legendre transform giving the transaction costs associated to the transactions on the order book at time $t$. The major insight of  \cite{CarmonaWebster_FS} was to prove that if at a scale given by an integer $N\ge 1$ we assume that the mid-price evolves according to the discretization $p^N_n =p_{n/N}$ for $n=0,1,\cdots,N$ of the continuous time $p_t$ (and similarly for the inventory $L^N_n=L_{n/N}$, then if the shape of the order book is give at the scale $N$ by:
\begin{equation}
\label{fo:gamma_normalization}
\gamma^N_n(x) = \frac{1}{N} \gamma_{n/N}\left(\sqrt{N} x\right).
\end{equation}
then if the single trade self-financing condition holds at scale $N$, then it is proven in \cite{CarmonaWebster_FS}  that the continuous time relationship between the liquidity provider wealth $X$, inventory $L$, the price $p$ and the transaction cost function $c$ is:
\begin{equation}
		dX_t = L_t dp_t + \Phi_{l_t}(c_t) dt + d[L,p]_t
\end{equation}
where $X_t = \lim_{N\rightarrow\infty} X^N_{\lfloor Nt\rfloor}$ u.c.p. In this formula, $\Phi_{\sigma^2}$ is the cumulative distribution function of the 
Gaussian distribution with mean $0$ and variance $\sigma^2$.

\vskip 6pt
The goal of this note is to present a couple of applications already touted in the literature, but which get a nw lease on life in light of the new self-financing condition. 

\begin{remark}
Before we proceed, we note that it is possible to recover the classical Almgren-Chriss price impact model, as well as the standard proportional transaction cost model, by choosing different renormalizations of the order book shape function different from \eqref{fo:gamma_normalization}. We refer the interested reader to Section 3 of \cite{CarmonaWebster_FS} for a detailed discussion of this important remark.
\end{remark}

\section{Price Impact and Models}
\label{sec:price_impact}
The above self-financing equations can be considered as  \emph{bare bones} descriptions of the market. They provide for an accountant's perspective on the market. Given a trader's inventory and the limit order book he or she trades on, the accountant can track his or her wealth perfectly. The number of hypotheses made are minimal in order to obtain the result which  tracks perfectly wealth once a trading strategy is given.

However, the above framework does not shed much light on which trading strategies can lead to the inventory processes 
satisfying the self-financing conditions we derived. Clearly, if any strategy is permissible in some reasonable sense, trading with limit orders is always preferable to trading with market orders as it is obviously preferable to capture the transaction costs rather than pay them. But in practice there is a trade-off between using limit and market orders. This trade-off has been captured by key words such as \emph{adverse selection}, \emph{price impact} and \emph{market response function}. These three terms are related, at least on an informal level, and correspond to attempts by different communities (the economics, mathematical finance and econophysics communities respectively) to model the hidden cost of placing limit orders. We propose our own approach within the framework of the previous two sections. For the sake of definiteness, we use the terminology \emph{price impact}, in line with the rest of the mathematical finance community.

We would like to stress that we regard Definition \ref{adv_selection} as the main thrust of this section, and our main contribution to the literature on price impact. It is motivated by the empirical fact that, with very high probability, $\Delta L \Delta p \ge 0 $ when trading with market orders and $\Delta L \Delta p \le 0$ when trading with limit orders. In plain words, these inequalities state that the price never moves immediately against a market order. It is a very simple, yet robust model of price impact.

\begin{definition}\label{adv_selection}
Let $p$ be a price process, $c$ a transaction cost process and $L$ an inventory process obtained exclusively via the use of limit orders. We say that the triplet $(p,c,L)$ is consistent if the sample paths of $\left[p,L\right]_t$ are strictly decreasing a.s. .
\end{definition}
We shall sometimes say that the couple $(p,L)$ is consistent when the transaction cost process $c$ 
corresponding to the order book on which the limit orders are placed is understood from the context.

\vskip 4pt
Next, we illustrate the significance of this definition on a specific model introduced to justify efforts such as \cite{Schied} and \cite{Wang} to model price impact as a deterministic mean-reversion of the limit order book between two trade times. The advantage of choosing such a model for the purpose of illustration is that we deal with bona fide  \emph{equations} rather than mere inequalities, leading to stronger results, albeit under stronger assumptions.

\subsection{More explicit and rigid price impact model}
The aim of this section is to show that a model for limit order fill rates with exact price recovery can lead to models where the price is a function of trade volumes, and vice-versa. This provides a model of supply and demand in high frequency markets, and closes the loop of our excursion in modeling. While this structural model is more rigid than our previous reduced form models, we believe that it illustrates some important market features, and provides yet another example of our micro-to-macro transition. By exact price recovery, we mean that in this model, one can compute the exact price move taking place when a limit order depletes liquidity on the limit order book, and otherwise, assume that the price recovers to a deterministic value between the previous price and this new price resulting from the move.

\subsubsection{Microscopic assumptions}
Let $(\Omega, \F, \P)$ be a probability space and $p$ and $L$ be two discrete time processes representing the market price and the inventory of a liquidity provider respectively. Let $\gamma$ be a $C^3$-function valued discrete time process representing our provider's shape function and $c$ its associated transaction cost process.
Our fundamental assumption is that
\begin{equation}
\Delta p = \lambda c'(-\Delta L)  \label{eq:supply_demand_p}
\end{equation}
or, equivalently
\begin{equation}
\Delta L = -\gamma'(\lambda^{-1} \Delta p) \label{eq:supply_demand_L}
\end{equation}
where $\lambda \in(0,1]$ is a real number that encapsulates price recovery. The bigger $\lambda$, the smaller the price recovery.

\subsubsection{Tools}
Equation \eqref{eq:supply_demand_p} allows a liquidity provider to derive the price from trade volumes and the order book, while equation \eqref{eq:supply_demand_L} derives the trade volumes from the prices and the order book. Both lead to the same consistency relationships between $p$, $L$ and $\gamma$ in the continuous limit.

Our analysis is based on a result from \cite{Jacod} which we state for the sake of completeness.

Let $(\Omega, \F, \FF, \P)$ be a filtered probability space supporting an $1$-dimensional $\FF$-Wiener process $W$, and $Y$ a $1$-dimensional  It\^o process of the form:
\begin{equation}
Y_t = Y_0 + \int_0^t b_t dt + \int_0^t \sigma_t dW_t, \qquad\qquad t\in[0,1].
\end{equation}

\begin{hypothesis}{((H)+ (K) from \cite{Jacod})}
\label{J1}
We assume that $b_t$ and $\sigma_t$ are progressively measurable, $b_t$ is locally bounded and $\sigma_t$ is c\`adl\`ag.
\end{hypothesis}

Now let $F : \Omega \times [0,1] \times \R \rightarrow \R$ be a random, $\FF$-adapted function that is $C^1$ in $y$ and $C^0$ in (t,y). We will shorten the notation to $y\mapsto F_t(y)$. 

\begin{hypothesis}{((7.2.1), (10.3.2), (10.3.3), (10.3.4) and (10.3.7) from \cite{Jacod})}\label{J3}
We assume that a.s. for all $t$, $F_t$ is an \emph{odd} function.
Furthermore, we assume that there exists a function $g: \R \rightarrow \R$ with at most polynomial growth, and a real number $\beta > 1/2$ such that, for all $\omega \in \Omega$, $(t,s)\in[0,1]^2$ and $y \in \R$ we have:
\begin{align*}
|F_t(y)| &\le g(y) \\
|F'_t(y)| &\le g(y) \\
|F_t(y) - F_s(y)| &\le g(y)|t-s|^\beta 
\end{align*}
\end{hypothesis}
We will make use of the following result (10.3.2) from \cite{Jacod}: 
\begin{theorem}\label{thm_J2}
Under assumptions \ref{J1} and \ref{J3}, there exists a very good filtered extension of the original space such that we have the following stable convergence in law as $N\rightarrow \infty$:
\begin{align*}
\frac{1}{\sqrt{N}}\sum_{n=1}^{\lfloor Nt \rfloor}& F_{n/N}\left({\sqrt{N}(Y_{(n+1)/N} - Y_{n/N})}\right) \rightarrow U_t
\end{align*}
where 
\begin{equation}
U_t = \int_0^t  b_s \Phi_{\sigma_s}\left(F'_s\right) ds  + \int_0^t \sqrt{\Phi_{\sigma_s}\left((F_s)^2\right)}	 dW'_s
\end{equation}
where $W'_t$ is a  $d$-dimensional Wiener process such that
\begin{equation*}
[W',W]_t =  \int_0^t \frac{\Phi_{\sigma_s}\left( id \, F^k_s\right)}{\sigma_s \sqrt{\Phi_{\sigma_s}(F^k_s)^2}} ds
\end{equation*}
and $id$ is the identity function.
\end{theorem}

\subsubsection{Continuous time setup}
Let $(\Omega, \F, \FF, \P)$ be a filtered probability space supporting a $\FF$-Wiener process $W$. We will fix either an It\^o process
\begin{equation}
p_t = p_0 + \int_0^t \mu_s ds + \int_0^t \sigma_s dW_s
\end{equation}
for the price or 
\begin{equation}
L_t = L_0 + \int_0^t b_s ds + \int_0^t l_s dW_s
\end{equation}
for the inventory. In addition to one of these processes, we also fix an order book shape process $\gamma=(\gamma_t)_{t\ge 0}$ and denote by $c=(c_t)_{t\ge 0}$ the associated transaction cost process. 

\vskip 2pt
Let us assume that $L$ (respectively $p$) satisfies Assumption \ref{J1} and $c'$ (respectively $\gamma'$) satisfies Assumption \ref{J3}. These assumptions basically state that the stochastic transaction cost function $c_t(\cdot)$ is symmetric, twice differentiable in space and has strong smoothness properties in time.

As before, we define the discretized processes $L^N_n = L_{n/N}$ (respectively $p^N_n = p_{n/N}$) and $c^N_n(\cdot) = \frac{1}{N} c_{n/N}\left(\sqrt{N} \cdot\right)$ (respectively $\gamma^N_n(\cdot) = \frac{1}{N} \gamma_{n/N}\left(\sqrt{N} \cdot\right)$).

\subsubsection{Main result}
The main result is a straightforward application of Theorem \ref{thm_J2}. If we are given the inventory $L$ and transaction costs $c$ then we have:

\begin{corollary}
There exists a very good filtered extension of the original space such that we have the stable convergence in law $p^N_{\lfloor Nt\rfloor} \rightarrow p_t$ with
\begin{equation}
dp_t = -\lambda b_t \Phi_{l_t}\left(c''_t\right) dt + \lambda \sqrt{\Phi_{l_t}((c'_t)^2)} dW'_t
\end{equation}
where 
\begin{equation}
[W',W]_t = -\int_0^t \frac{\Phi_{l_s}\left(id \, c'_s\right)}{l_s \sqrt{\Phi_{l_s}((c'_s)^2)}}ds.
\end{equation}
In particular, 
\begin{equation}
d[p,L]_t = -\Phi_{l_t}\left(id \, c'_t\right)dt
\end{equation}
\end{corollary}
A completely analogous result is obtained if we start with the price $p$ and an order book shape function $\gamma$:

\begin{corollary}
There exists a very good filtered extension of the original space such that we have the stable convergence in law $L^N_{\lfloor Nt\rfloor} \rightarrow L_t$ with
\begin{equation}
dL_t =  -\mu_t \Phi_{\sigma_t}\left(\gamma''_t(\lambda^{-1} \cdot)\right) dt + \sqrt{\Phi_{\sigma_t}((\gamma'_t)^2(\lambda^{-1} \cdot))} dW'_t
\end{equation}
where
\begin{equation}
d[p,L]_t = -\Phi_{\sigma_s}\left(id \, \gamma'_t(\lambda^{-1} \cdot)\right)dt.
\end{equation}
\end{corollary}

\subsection{A special case}
While quite unrealistic, the model of a flat order book has been used repeatedly because of its extreme tractability. See for example \cite{Alfonsi, Wang}. It corresponds to a shape function satisfying $\gamma''_t = m_t$ for some real valued adapted process $m=(m_t)_{t\ge 0}$ taking only positive values. 
The corresponding cost function is hence quadratic, and this quadratic transaction cost model leads to the \emph{linear} price impact model:
\begin{equation}
\label{fo:special_case}
	\begin{cases}
		dp_t &= -\frac{\lambda}{m_t} dL_t \\
		dX_t &= L_t dp_t + \left(\frac{1}{2} - \lambda\right) \frac{l^2_t}{m_t} dt.
	\end{cases}
\end{equation}
Note that the sign of the effective transaction costs is that of $\frac{1}{2} - \lambda$. Indeed, in the self-financing case $\lambda = \frac{1}{2}$, and price recovery and price impact perfectly cancel each other. If $\lambda > \frac{1}{2}$, then the price impact of trades is stronger than the collected spread because of insufficient price recovery. Also, because of the uniform structure of the order book and perfect fill rate, the inventory of the provider is perfectly anti-correlated with the price.

\subsubsection*{Absence of price manipulation strategies}

A major concern for any dynamic model of market microstructure is the possible existence of price manipulation strategies. These can be defined in multiple ways and have been studied extensively by Schied et al in \cite{Schied} among others. In the present context, we define price manipulation in the following way:

\begin{definition}[Price manipulation]
Let $c$ be a transaction cost function, $\lambda$ a price recovery parameter, and $\mathcal{A}$ a set of couples of processes $(p,L)$ which are consistent (with respect to $c$) in the sense of Definition \ref{adv_selection}. We say that this set $\mathcal{A}$ is subject to price manipulation if there exists a $(p,L)\in \mathcal{A}$ such that $L_1 = L_0$ and
\begin{equation}
\E[X_1] > \E[X_0].
\end{equation}
\end{definition}
In words, we want to exclude round-trip statistical arbitrages, given that the liquidity provider does not control the incoming market orders.
The tractability of the flat order book allows us to rule out price manipulation strategies in many situations.

\begin{proposition}
\label{pr:pms}
Consider a market with price recovery parameter $\lambda$. A constant flat order book, $\gamma''_t \equiv c$ for some $c>0$ does not allow for price manipulation if $\lambda \ge 1$. On the other end, price manipulation strategies do exist when $\lambda < 1$.
\end{proposition}

\begin{proof}
Using both equations of \eqref{fo:special_case} we get:
\begin{equation*}
dX_t  =-\frac{\lambda}{2 m_t}d (L^2)_t + \frac{1-\lambda}{2 m_t}l^2_t dt
\end{equation*}
so that integrating between $t=0$ and $t=1$ we get, if $L_1=L_0$:
\begin{equation}
X_1 = X_0 + \frac{1-\lambda}{2}\int_0^1 \frac{l_s^2}{m_s} ds
\end{equation}
from which one easily concludes.
\end{proof}

\section{Applications}\label{sec:applications}
Applications of the proposed relationships depend on models of the inventory and price processes $L$ and $p$. In the sequel, when we formulate an optimization problem, we assume that the inventory can be any It\^o process. This is an act of faith as making it happen typically requires good execution algorithms and limit order fill rates. In any case, to be consistent with the results of \cite{CarmonaWebster_FS} recalled above, we require that these processes satisfy $d\left[p,L\right]_t<0$ when using limit orders and $d\left[p,L\right]_t \ge 0$ for market orders.

\subsection{Option Hedging}

In this section we derive a pricing PDE for a European option in a local volatility model with transaction costs and price impact as given by the adverse selection term in our self-financing condition.
We highlight the consequences of trading with market orders only as opposed to trading with limit orders. This is directly related to the literature on option hedging under gamma constraint. We recover the same PDE structure and interpretation of the gamma penalization as a liquidity cost, as in \cite{Soner} and \cite{Touzi}, albeit through a different path. 
These papers do not emphasize adverse selection and no discussion of the order book are given.
The first paper starts from the gamma-penalized Partial Differential Equation (PDE) while the other derives it as a limit of PDE models. Neither paper derives the non-linear penalty in the PDE from a microstructure model. They are more focused on what the penalty term implies.

In addition to presenting a microstructure-based approach to the problem, our solution also answers a very practical-minded question: \emph{Should one delta-hedge with limit orders or market orders?} We argue that the answer depends upon the sign of the gamma of the option.

In what follows, we first treat the fully non-linear case to exhibit the strong parallel with \cite{Touzi} and \cite{Soner}. The linear case follows as a corollary. Clearly, it is of more practical interest as a tractable extension of the Black and Scholes framework.

\subsubsection{Mathematical setup}
Let $(\Omega, \F, \FF, \P)$ be a filtered probability space and $W$ be a $\FF$-Brownian motion which generates the filtration $\FF$. We assume that the mid-price $p_t$ is given exogenously and satisfies the stochastic differential equation (SDE):
\begin{equation}
dp_t = \mu(p_t) dt + \sigma(p_t) dW_t,
\end{equation}
where $\mu$ and $\sigma$ are two globally Lipschitz functions. We also assume that the order book shape function process $\gamma=(\gamma_t)_{t\ge 0}$ is continuous and random only through the price level $p$ in the sense that $\gamma_t(\alpha) = \gamma(p_t,\alpha)$ for a deterministic function $(p,\alpha)\hookrightarrow \gamma(p,\alpha)$. We single out a trader, and define an inventory process as an $\FF$-adapted It\^o process 
\begin{equation}
L_t = L_0 + \int_0^t b_u du + \int_0^t l_u dW_u
\end{equation}
where $(b_t)_{t\ge 0}$ and $(l_t)_{t\ge 0}$ are $\FF$-adapted, c\`adl\`ag processes. Note that in this formulation, $l_t$ is \emph{signed}. Since we identified trading with limit or market orders to the sign of $d[p,L]_t$, we shall impose the constraint $l_t < 0$ when we model trading with limit orders, and $l_t \ge 0$ when trading is done with market orders.

Denote by $c(p,\cdot)$ the Legendre transform of $\gamma(p,\cdot)$. We will use the function
\begin{equation}
g(p,l) = \text{sign}(l)\Phi_{l}(c(p,\cdot))
\end{equation}
Given a real number $K_0$ representing the trader's initial cash endowment, according to our self-financing condition, her wealth is given by:
\begin{equation}
\label{fo:wealth}
X_t = L_0 p_0 + K_0 + \int_0^t L_u dp_u + \int_0^t \left(\sigma(p_u) l_u - g(p_u,l_u) \right)du.
\end{equation}
This follows from the self-financing equation proved in Proposition \ref{prop:self_financing_continuous}.

Let $f \in C^0$ be  the payoff function of a European option with maturity $T$. We define a perfect replication strategy as follows.
\begin{definition}
An initial cash endowment $K_0$ and an inventory process $L=(L_t)_{t\ge 0}$ are said to perfectly replicate the European payoff $f(p_T)$ at maturity $T$ if
\begin{equation}
X_T = f(p_T)
\end{equation}
and the corresponding replication price is defined as $X_0 = K_0 + p_0 L_0$.
\end{definition}

\subsubsection{The result}

\begin{theorem}
Let $f\in C^0$ and $T>0$. Assume that $v\in C^{1,3}$ solves the PDE:
\begin{equation}
\label{fo:bs_pde}
\frac{\partial v}{\partial t}(t,p) + g\left(p, \sigma(p) \frac{\partial^2 v}{\partial p^2}(t,p)\right) - \frac{\sigma^2(p)}{2}\frac{\partial^2 v}{\partial p^2}(t,p) = 0
\end{equation}
with terminal condition $v(T,p) = f(p)$. Then  $L_t = \frac{\partial v}{\partial p}(t,p_t)$, and $K_0 = v(0,p_0) - \frac{\partial v}{\partial p}(0,p_0) p_0$ form a perfect replication strategy for the payoff $f(p_T)$ at maturity $T$. Furthermore,  the volatility of the replicating inventory is given by
\begin{equation}
l_t = \sigma(p_t) \frac{\partial^2 v}{\partial p^2}(t,p_t),
\end{equation}
and the replication price of the option is $X_0 = v(0,p_0)$.
\end{theorem}

\begin{proof}
Choosing $L_t = \frac{\partial v}{\partial p}(t,p_t)$ leads to
\begin{equation}
l_t = \sigma(p_t) \frac{\partial^2 v}{\partial p^2}(t,p_t)
\end{equation}
and
\begin{equation}
b_t = \frac{\partial^2{v}}{\partial t \partial p}(t,p_t) + \mu(p_t)\frac{\partial^2{v}}{\partial p^2}(t,p_t) +  \frac{\sigma^2(p_t)}{2}\frac{\partial^3{v}}{\partial p^3}(t,p_t).
\end{equation}
As $v \in C^{1,3}$, this choice of $L_t$ is therefore a bona fide inventory process. Writing down the dynamics of the wealth leads to:
\begin{align*}
dX_t &= \left(\sigma(p_t) l_t - g(p_t,l_t) \right) dt + L_t dp_t \\
						 &=  \left(- g(p_t,\sigma(p_t) \frac{\partial^2 v}{\partial p^2}(t,p_t)) + \sigma^2(p_t)\frac{\partial^2 v}{\partial p^2}(t,p_t)\right) dt  + \frac{\partial v}{\partial p}(t,p_t) dp_t \\
						 &= \left( \frac{1}{2}\sigma^2(p_t)\frac{\partial^2 v}{\partial p^2}(t,p_t) + \frac{\partial v}{\partial t}(t,p_t)\right)dt  + \frac{\partial v}{\partial p}(t,p_t) dp_t \\
						 &= d(v(t,p_t)).
\end{align*}
Since the initial values match, we have that $X_t = v(t,p_t)$ for all times. This concludes.
\end{proof}

We recall that an option is said to have positive gamma when $\frac{\partial^2 v}{\partial p^2}(t,p) >0$ for all $t$ and $p$. A negative gamma option is one for which $\frac{\partial^2 v}{\partial p^2}(t,p) <0$ for all $t$ and $p$.

\begin{corollary}
When using only one type of order (say limit or market orders), positive gamma options can only be hedged (in the sense that their payoffs can be replicated) with market orders. Negative gamma options can only be hedged with limit orders.
\end{corollary}
\begin{proof}
The identity 
\begin{equation}
l_t = \sigma(p_t) \frac{\partial^2 v}{\partial p^2}(t,p_t)
\end{equation}
implies that $l_t$ and the option gamma must be of the same sign.
\end{proof}

\begin{remark}
We illustrate the above result by commenting on a simple practical example. If one buys and delta-hedges a call option, then a synthetic negative gamma position must be created through a dynamic trading strategy. Because the delta of a call option increases as the price increases, the trader must sell when the price goes up, and buy when the price goes down. Given that limit orders tend to buy when the price goes up and sell when the price down because of price impact, their use to lower the cost of delta-hedging makes sense. 
\end{remark}

\begin{remark}
The above result leads to an intuitive implementation strategy for delta-hedging negative gamma positions. By computing the gamma of an option, one can figure out how much needs to be bought should the price move down, or sold should the price move up. By placing limit orders \emph{at the end of the queue}, it is impossible for the price to move without our hedge getting executed. Since the model has only one source of random shocks (the Wiener process $W=(W_t)_{t\ge 0}$), it is expected to be complete in the sense that perfect replication holds. In practical situations, the incompleteness of the markets and the fact that the price could move back introduce some hedging risk. Nevertheless, the negative correlation created by trading with limit orders lowers the cost of the delta-hedge.
\end{remark}

The PDE \eqref{fo:bs_pde} is non-linear in the second derivative of $v$. However, it is linear when all the orders can be filled at the pre-trade best-bid or best-ask prices, in which case we recover a Black-Scholes type formula.

\begin{corollary}
Assume that all the orders can be filled at the pre-trade best-bid or best-ask prices and denote by $s_t = s(p_t)$ the bid-ask spread. Then the pricing PDE \eqref{fo:bs_pde} becomes
\begin{equation}
\frac{\partial v}{\partial t}(t,p) + \left(\frac{\sigma(p)s(p)}{\sqrt{2\pi}} - \frac{\sigma^2(p)}{2}\right)\frac{\partial^2 v}{\partial p^2}(t,p) = 0. 
\end{equation}
\end{corollary}
\begin{proof}
In the bid-ask spread case, $c(p,l) = \frac{s(p)}{2}|l|$ and hence $g(p,l) = \frac{s(p)}{\sqrt{2\pi}} l$.
\end{proof}

In particular, we recover the standard, frictionless local volatility model when $s(p) = \sqrt{2 \pi} \sigma(p)$. This is because in our self-financing condition, the terms representing transaction costs and adverse selection (i.e. price impact) exactly cancel each other out, and the frictionless self-financing equation holds true.
\vskip 2pt
In addition to computing the price and delta-hedging ratios under transaction costs and instantaneous adverse selection, this theory suggests an execution strategy by specifying when limit or market orders should be used to hedge an option.

\subsection{Market Making}

For our second application, we adapt to our framework the key insight of the model proposed in \cite{Stoikov}. The aim is to solve the optimization problem of a representative market maker controlling the spread and maximizing her profits. The trade-off she faces, and which is the key ingredient of the model, is the following: the smaller the spread, the likelier trades are, but the less profit she makes on each of them.

In a way similar to \cite{Stoikov, Stoikov2}, we model the probability of execution of a limit order by a decreasing function of the quoted spread. This will first be done at the microscopic level, to obtain a reasonable model for our inventory process $L$ at the macroscopic level. A key difference with \cite{Stoikov} is that we still impose the price impact constraint, which further depresses the market maker's profits because of adverse selection. In this respect, our model is closer to that proposed in \cite{Cartea}, which also proposes a market making control problem subject to adverse selection.

To make sure that the price impact constraint is satisfied, we use, at the microscopic level, a modified version of the Almgren and Chriss model \cite{Almgren} to relate the change in mid-price to the change in the aggregate inventory of the liquidity providers.
We assume that
\begin{equation}
\Delta_n L = - \lambda_{n+1} \Delta_n p
\end{equation}
for a $\F_{n+1}$-measurable, positive random variable $\lambda_{n+1}$. This is an non-predictable form of linear price impact, in the sense that, ex-post, the price increment is a linear function of the traded volume.

To capture the insight of \cite{Stoikov}, we model $\lambda_{n+1}$ in such a way that
\begin{equation}
\E[\left.\lambda_{n+1}\right|\F_n] = \rho_n(s_n) f_n(s_n) \quad\text{and}\quad \E[\left.\lambda^2_{n+1}\right|\F_n] = \left(f_n(s_n)\right)^2
\end{equation}
where $s_n$ is the market maker's chosen spread, and $\rho_n$ and $f_n$ are continuous, positive function with $f_n$ decreasing and $\rho_n \in [0,1]$. The assumption that $f_n$ is decreasing in the spread is inherited from \cite{Stoikov}, and the fact that $\rho_n$ must be smaller than $1$ is due to Jensen's convexity inequality. We assume that $\lambda_{n+1}$ is independent of $\Delta_n p$ conditionally on $\F_n$.
Computing the predictable quadratic variation of $L_n$ yields:
\begin{equation}
 \sum_{k=1}^{n-1} f^2_k(s_k) \E\left[\left.\Delta_k p^2\right|\F_k\right],  
\end{equation}
while the predictable quadratic covariation of $L_n$ and $p_n$ is given by:
\begin{equation}
- \sum_{k=1}^{n-1} \rho_k(s_k) f_k(s_k) \E\left[\left.\Delta_k p^2\right|\F_k\right].  
\end{equation}
This suggests the use of the following model in the continuum limit:
\begin{equation}
\begin{cases}
dp_t &= \mu_t dt + \sigma_t dW_t \\
dL_t &= -\rho_t(s_t) f_t(s_t) \mu_t dt + f_t(s_t)\sigma_t dW'_t 
\end{cases} 
\end{equation}
with $d[W,W']_t = -\int_0^t \rho_u(s_u) du$ for some adapted, continuous  and positive functions $\rho_t(\cdot)$ and $f_t(\cdot)$ with $\rho_t \le 1$ and $f_t$ decreasing.
Note that the equation for $L_t$ can also be written as:
\begin{equation}
dL_t = -\rho_t(s_t) f_t(s_t) dp_t + f_t(s_t)\sqrt{1 - \rho^2_t(s_t)}\sigma_t dW^\perp_t
\end{equation}
with a Wiener process $W^\perp_t$ independent from $W_t$. We will from now on assume that $p_t$ is adapted to the filtration generated by $W_t$.

Applying our wealth equation, we obtain:
\begin{equation}
X_T = L_T p_T - \int_0^T p_t dL_t + \frac{1}{\sqrt{2\pi}} \int_0^T \sigma_t  s_t f_t(s_t)dt .
\end{equation}
For both $f_t$ and $\rho_t$, a natural assumption is that they are time-independent functions of the spread rescaled by the volatility:
\begin{equation}
f_t(s) = f(s/\sigma_t); \quad \rho_t(s_t) = \rho(s_t/\sigma_t)
\end{equation}
for some $C^0$ decreasing function $f$ and $C^0$ function $\rho$. We will furthermore assume that $g(x) = xf(x)$ is a decreasing function for $x$ large enough, that $g(x) \rightarrow 0$ as $x\rightarrow \infty$, and that $f(x)>0$ for all $x\ge 0$.

The problem of a risk-neutral market maker attempting to set the spread optimally is to maximize:
\begin{equation}
\sup_{s} \E X_T.
\end{equation}
This is a classical stochastic control problem which we solve using the Pontryagin maximum principle. Let us define a few functions first.

\begin{lemma}
For all $a>0$, define the function $F_a$ by
\begin{equation}
F_a : x \mapsto \frac{x}{\sqrt{2\pi}} f(x) - a \rho(x) f(x)
\end{equation}
Then the function 
\begin{equation}
M(a) = \max_{x \in [0,\infty)} F_a(x) \label{def_M}
\end{equation}
is well defined, continuous, and decreasing in $a$. Furthermore, there exist a measurable selection
\begin{equation}
m(a) \in \textit{argmax}_{x \in [0,\infty)} F_a(x)
\end{equation}
and we have that $m(a) >0$.
\end{lemma} 
\begin{proof}
First, note that for all $a>0$, 
$$
F_a(0) = - a \rho(0) f(0) \le 0, \quad F_a(a+1) \ge f(a+1) >0
$$
Next, if $g$ is decreasing on the interval $[x_0, \infty)$, then we can define the function $\beta(a)$ as
$g^{-1} \circ f(a+1)$ if $f(a+1)$ is in $g[x_0,\infty)$, and $x_0$ otherwise. $\beta(a)$ is continuous and satisfies $f(a+1) \ge g(x)$ for all $x\in (\beta(a), \infty)$.

This proves that the maximum of $F_a$ is attained on the compact interval $[a+1, \beta(a)]$. The continuity of $M$ holds by Berge's maximum theorem. It is decreasing by definition of $F_a$. The measurable selection result follows by Theorem 18.19 of \cite{Aliprantis}.
\end{proof}

\begin{proposition}
Any solution of the control problem is of the form
\begin{equation}
\frac{s_t}{\sigma_t} = m\left(\alpha_t\right)
\end{equation}
where 
\begin{equation}
\alpha_t = \E\left[\left. p_T - p_t\right|\F_t\right] \frac{\mu_t}{\sigma^2_t} + \frac{Z_t}{\sigma_t},\label{def_alpha}
\end{equation}
$Z_t$ being the volatility of the martingale representation of $p_T$
\end{proposition}
\begin{proof}
We apply the necessary part of the stochastic Pontryagin maximum principle.
The generalized Hamiltonian is equal to:
\begin{eqnarray*}
&&\mathcal{H}_t(s, L, Y, Z, Z^\perp) = - \rho(s/\sigma_t) f(s/\sigma_t)\left[\left(Y_t -p_t\right)\mu_t + \sigma_t Z\right] \\
&&\phantom{???????????????????}+ \frac{\sigma_t}{\sqrt{2 \pi}} s f(s/\sigma_t) +  \sigma_t f(s/\sigma_t) \sqrt{1 - \rho^2(s/\sigma_t)} Z^{\perp}
\end{eqnarray*}
and the adjoint equation is solved by 
\begin{equation}
Y_t = \E\left[\left. p_T \right|\F_t\right]
\end{equation}
which, in particular, implies $Z^\perp_t =0$. $Z_t$ can be computed via the martingale representation theorem on the Brownian filtration generated by $W_t$.

The Hamiltonian to maximize therefore becomes
\begin{equation}
\sigma^2_t F_{\alpha_t} \left(\frac{s}{\sigma_t}\right)
\end{equation}
and the previous lemma concludes.
\end{proof}

Beyond the optimal control, one might be interested in the dependence in $\sigma_t$ and $\alpha_t$ of the market maker's expected profit as well as the volatility of her inventory. Note that a low volatility of the inventory means that the market maker has essentially pulled out of the market.

\begin{corollary}
The market maker's expected profit is
\begin{equation}
\E\left[\int_0^T M\left(\alpha_t\right) \sigma^2_t dt\right]
\end{equation}
where $\alpha_t$ is given by \eqref{def_alpha} and the function $M$ is defined by \eqref{def_M}. The volatility of her inventory is 
\begin{equation}
\sigma_t f (m\left(\alpha_t\right)).
\end{equation}
\end{corollary}

\begin{proof}
The expected profit can be computed by integrating the Hamiltonian along the optimal path. The rest follows from the previous proposition.
\end{proof}

Recall that $M$ is a decreasing function. A consequence of the corollary is that the market maker's expected profit is a decreasing function of $\alpha_t$ and, for $\alpha_t$ being fixed, an increasing function of the volatility.

Two issues need to be resolved if one looks for tractable formulas. First, an explicit model for $p_T$ must be given for which the martingale representation term $Z_t$ can be computed. Second, one has to propose a function $g$ for which the maximal argument $m$ of $F$ can easily be characterized as a function of $\alpha_t$. We address them in some specific situations.

\subsubsection{The martingale case}
Note that the latter problem is easily solved when $p_t$ is assumed to be a martingale. Indeed, if we have
\begin{equation}
dp_t = \sigma_t dW_t
\end{equation}
for some adapted, continuous and positive process $\sigma_t$, then $\alpha_t = 1$ and we simply have 
\begin{equation}
s_t = m(1) \sigma_t
\end{equation}
circumventing the need for explicit functions $\rho$ and $f$. This result provides a theoretical foundation for the empirical claim made in \cite{Bouchaud2} that the spread is a linear function of volatility.

Plugging this optimal spread back into the objective function, the market maker's expected profit and loss (P\&L) is given by:
\begin{equation}
M(1) \E\left[\int_0^T \sigma^2_t dt\right].
\end{equation}

Given that the dynamics of the market maker's inventory are:
\begin{equation}
dL_t = -\rho_t(s_t) f_t(s_t) dp_t + f_t(s_t)\sqrt{1 - \rho^2_t(s_t)}\sigma_t dW^\perp_t,
\end{equation}
we conclude that the inventory is also a martingale. In option-pricing terms, $-\rho_t(s_t) f_t(s_t)$ is the Gamma exposure of the market maker: it measures the changes in inventory due to a change in the price. It is negative. The expected profit and loss however, exhibits positive Vega as it is an increasing function of volatility.

\subsubsection{Explicit cases}
Other cases where $\alpha_t$ can be computed explicitly are:
\begin{itemize}
\item the Black-Scholes model
\begin{equation}
dp_t = \mu p_t dt + \sigma p_t dW_t
\end{equation}
in which case we obtain:
\begin{equation}
\E \left[\left. p_T\right|\F_t\right] = p_t e^{\mu(T-t)}; \quad Z_t = \sigma p_t e^{\mu(T-t)},
\end{equation}
and hence 
\begin{equation}
\alpha_t = \frac{\mu}{\sigma^2} \left(e^{\mu(T-t)} -1\right) + e^{\mu(T-t)}.
\end{equation}

\item the case of a mean reverting (Ornstein-Uhlenbeck) price process 
\begin{equation}
dp_t = \rho (p_0 - p_t) dt + \sigma dW_t
\end{equation}
in which case: 
\begin{equation}
\E \left[\left. p_T\right|\F_t\right] = p_0 + e^{-\rho (T- t)}(p_t - p_0); \quad Z_t = \sigma e^{- \rho(T - t)},
\end{equation}
and hence
\begin{equation}
\alpha_t =   - \frac{\rho}{\sigma^2} \left(p_t - p_0\right)^2\left(e^{-\rho(T-t)} -1\right) + e^{-\rho(T-t)}.
\end{equation}
\end{itemize}

Unlike in the martingale case, it is hard to obtain any tractable formulas without specifying a functional form for $\rho$ and $f$. In the case where $\rho(x) = 1/(1+x)$ and $f(x) = 1/(1+x)^2$, the optimal spread becomes
\begin{equation}
s_t = \sigma_t \sqrt{1 + 3 \alpha_t}
\end{equation}

Note that $m$ is an increasing function of $\alpha_t$. To compare with the martingale case, where $\alpha_t = 1$, we therefore want to compare the ratio of $\alpha_t$ to $1$ in order to study the impact of the model assumptions on the market maker's profit and inventory volatility.

\begin{itemize}
\item For the Black-Scholes model, $\alpha_t$ is larger than $1$ for $\mu >0$. For $\mu<0$, there exists a critical value depending on $T$ and $\sigma$ for which this ratio flips sign.

\item In the case of an Ornstein-Uhlenbeck process, $\alpha_t$ is smaller than $1$ iff
\begin{equation}
\left(p_t - p_0\right)^2 < \frac{\sigma^2}{\rho} \label{significantly_away}
\end{equation}
that is, if the current price $p_t$ isn't too far from the long-term average $p_0$.
\end{itemize}
\vskip 4pt\noindent

In line with intuition, the market maker quotes larger spreads, expects less profit, and captures less volume in the 'momentum' Black-Scholes model, as compared to the martingale case. This is because the spread is an increasing function of $\alpha_t$, the expected profit a decreasing function of $\alpha_t$ and the volatility of the market maker's inventory is a decreasing function of the spread and hence of $\alpha_t$. In a mean-reverting market, unless the price is significantly away from its long-term trend as measured by the inequality \eqref{significantly_away}, the market maker quotes smaller spreads, expects more profit and captures more volume than in the two other market models.


\begin{thebibliography}{10}

\bibitem{Yacine}
Y.~Ait-Sahalia and J.~Jacod.
\newblock {\em High-Frequency Financial Econometrics}.
\newblock Princeton University Press, 2014.

\bibitem{Alfonsi}
A.~Alfonsi, A.~Fruth, and A.~Schied.
\newblock Optimal execution strategies in limit order books with general shape
  functions.
\newblock {\em Quantitative Finance}, 10(2):143--157, 2010.

\bibitem{Schied}
A.~Alfonsi,  A.~Schied, and A.~Slynko.
\newblock Order book resilience, price manipulation, and the positive portfolio problem. 
\newblock {\em SIAM Journal on Financial Mathematics} , 3(1):511-533, 2012

\bibitem{Aliprantis}
C.~Aliprantis and K.~Border.
\newblock {\em Infinite Dimensional Analysis}.
\newblock Springer, 2006.

\bibitem{Almgren}
R.~Almgren and N.~Chriss.
\newblock Optimal execution of portfolio transactions.
\newblock {\em Journal of Risk}, 3(2):5--39, 2000.

\bibitem{Ane}
T.~An\'e and H.~Geman.
\newblock Order flow, transaction clock, and normality of asset returns.
\newblock {\em The Journal of Finance}, 55(5):2259--2284, 2000.

\bibitem{Stoikov}
M.~Avellaneda and S.~Stoikov.
\newblock High-frequency trading in a limit order book.
\newblock {\em Quantitative Finance}, 8(3):217--224, 2007.

\bibitem{Bouchaud}
J.~. Bouchaud, Y.~Gefen, M.~Potters, and M.~Wyart.
\newblock Fluctuations and response in financial markets: The subtle nature of
  'random' price changes.
\newblock {\em Quantitative Finance}, 4(2):176--190, 2004.

\bibitem{cw_self_financing}
R.~Carmona, and K.~Webster.
\newblock Trading frictions in high frequency markets.
\newblock {\em Princeton University Tech. Rep. }, 2014.
\newblock{arxiv, http://arxiv.org/abs/1709.02015}


\bibitem{CarmonaWebster_FS}
R.~Carmona, and K.~Webster.
\newblock The Self-Financing Equation in Limit Order Book Markets.
\newblock {\em Finance \& Stochastics}, 2019 (to appear).
\newblock{arxiv, http://arxiv.org/abs/??????}


\bibitem{Cartea}
A.~Cartea, and S.~Jaimungal.
\newblock Risk metrics and fine tuning of high frequency trading strategies.
\newblock {\em Mathematical Finance}, 25(3):576-611, 2015.

\bibitem{Chellathurai}
T.~Chellathurai and T.~Draviam.
\newblock Dynamic portfolio selection with nonlinear transaction costs.
\newblock {\em Proceedings of the Royal Society A: Mathematical, Physical and
  Engineering Sciences}, 461(2062):3183--3212, 2005.

\bibitem{Touzi}
U.~Cetin, H. Mete Soner, and N.~Touzi.
\newblock Options hedging for small investors under liquidity costs.
\newblock {\em Finance Stochastics}, 14(3): 317-341, 2010.

\bibitem{Cont3}
R.~Cont and A.~de~Larrard.
\newblock Order book dynamics in liquid markets: limit theorems and diffusion
  approximations.
\newblock {\em Working paper}, 2011.

\bibitem{Cont2}
R.~Cont and A.~de~Larrard.
\newblock Price dynamics in a markovian limit order book market.
\newblock {\em SIAM Journal for Financial Mathematics}, 4(1):1--25, 2013.

\bibitem{Cont}
R.~Cont, S.~Stoikov, and R.~Talreja.
\newblock A stochastic model for order book dynamics.
\newblock {\em Operations Research}, 58(3):549--563, 2010.

\bibitem{Jacod}
J.~Jacod and P.~Protter.
\newblock {\em Discretization of Processes}.
\newblock Springer, 2011.

\bibitem{MuhleKarbe}
R.~Liu and J.~Muhle-Karbe.
\newblock Portfolio choice with stochastic investment opportunities: a user's
  guide.
\newblock Proc. 1st Princeton Summer School in Mathematical
  Finance, 2013. http://arxiv.org/abs/1311.1715

\bibitem{Magill}
M.~Magill and G.~Constantinides.
\newblock Portfolio selection with transactions costs.
\newblock {\em Journal of Economic Theory}, 13(2):245 -- 263, 1976.

\bibitem{Wang}
A.~Obizhaeva and J.~Wang.
\newblock Optimal trading strategy and supply/demand dynamics.
\newblock {\em Journal of Financial Markets}, 16:1--32, 2013.

\bibitem{Shreve}
S.E. Shreve and H.M. Soner.
\newblock Optimal investment and consumption with transaction costs.
\newblock {\em Annals of Applied Probability}, 4(3):609--692, 1994.


\bibitem{Soner}
H. M. Soner, and N.~Touzi.
\newblock Hedging Under Gamma Constraints By Optimal Stopping and Face-Lifting.
\newblock {\em Mathematical Finance}, 17(1):59--79, 2007.

\bibitem{Stoikov2}
S.~Stoikov and M.~Saglam.
\newblock Option market making under inventory risk.
\newblock {\em Review of Derivatives Research}, 12(1):55--79, 2009.

\bibitem{Bouchaud2}
M.~Wyart, J.~. Bouchaud, J.~Kockelkoren, M.~Potters, and M.~Vettorazzo.
\newblock Relation between bid-ask spread, impact and volatility in
  order-driven markets.
\newblock {\em Quantitative Finance}, 8(1):41--57, 2008.

\end{thebibliography}
\end{document}